\documentclass[a4paper,10pt]{article}

\usepackage[utf8]{inputenc}     

\usepackage{amsfonts}
\usepackage{amsthm}
\usepackage{amssymb}
\usepackage{amsmath}
\usepackage{enumerate}
\usepackage{tikz}
\usepackage{url}
\usepackage{nameref} 
\usepackage{hyperref} 
\usepackage{float}

\newtheorem{theorem}{Theorem}

\newtheorem{lemma}[theorem]{Lemma}
\newtheorem{proposition}[theorem]{Proposition}

\definecolor{darkgreen}{RGB}{0,128,0}

\newcommand{\trans}[1]{\mathchoice{\xrightarrow{#1}}{\xrightarrow{\smash{\lower1pt\hbox{$\scriptstyle #1$}}}}{\text{Error}}{\text{Error}}}
\newcommand{\transducer}{\mathcal{T}}
\newcommand{\automaton}{\mathcal{A}}
\newcommand{\size}[1]{\lvert #1 \rvert}
\newcommand{\sizeSet}[1]{\# #1}
\newcommand{\freq}[0]{\operatorname{freq}}

\newcommand{\WG}[1]{:#1}
\newcommand{\weight}[0]{\operatorname{weight}}
\newcommand{\pEmpty}{p^\lambda}
\newcommand{\pOutput}{p^o}
\newcommand{\qEmpty}{q^\lambda}
\newcommand{\qOutput}{q^o}

\newcommand{\R}[1]{\textcolor{red}{#1}}
\newcommand{\A}[1]{\textcolor{blue}{#1}}
\newcommand{\V}[1]{\textcolor{darkgreen}{#1}}

\begin{document}
\title{Preservation of normality by transducers}

\author{Olivier Carton \and  Elisa Orduna}
\date{October 31, 2018}
\maketitle


\begin{abstract}
  We consider input-deterministic finite state transducers with infinite
  inputs and infinite outputs, and we consider the property of Borel
  normality on infinite words.  When these transducers are given by a
  strongly connected set of states, and when the input is a Borel normal
  sequence, the output is an infinite word such that every word has a
  frequency given by a weighted automaton over the rationals.  We prove
  that there is an algorithm that decides in cubic time whether an
  input-deterministic transducer preserves normality.
\end{abstract}


\section{Introduction}

We start with the definition of \emph{normality} for real numbers, given by
\'Emile Borel~\cite{Borel09} more than one hundred years ago. A real
number is normal to an integer base if, in its infinite expansion expressed
in that base, all blocks of digits of the same length have the same
limiting frequency.  Borel proved that almost all real numbers are normal
to all integer bases. However, very little is known on how to prove that a
given number has the property.

The definition of normality was the first step towards a definition of
randomness.  Normality formalizes the least requirements about a random
sequence.  It is indeed expected that in a random sequence, all blocks of
symbols with the same length occur with the same limiting frequency.
Normality, however, is a much weaker notion than the one of purely random
sequences defined by Martin-L\"of.

The motivation of this work is the study of transformations preserving
randomness, hence preserving normality.  The paper is focused on very
simple transformations, namely those that can be realized by finite-state
machines.  We consider input-deterministic automata with outputs, also
known as sequential transducers mapping infinite sequences of symbols to
infinite sequences of symbols.  The main result is that it can be decided
in cubic time whether such a machine preserves or not normality.
Preserving normality means that the output sequence is normal whenever the
input sequence is.

The main result is obtained through a second result involving weighted
automata.  This second result states that if a sequential transducer is
strongly connected, then the frequency of each block in the output of a run
with a normal input is given by a weighted automaton on rational numbers.
It implies, in particular, that the frequency of each block in the output
sequence does not depend on the input as long as this input sequence is
normal.

This is not the first result linking normality and automata.  A fundamental
theorem relates normality and finite automata: an infinite word is normal
to a given alphabet if and only if it cannot be compressed by lossless
finite transducers. These are deterministic finite automata with injective
input-output behaviour. This result was first obtained by joining a theorem
by Schnorr and Stimm~\cite{Schnorr71} with a theorem by Dai, Lathrop, Lutz
and Mayordomo~\cite{Dai04}. Becher and Heiber gave a direct
proof~\cite{BecherHeiber13}.

Agafonov's Theorem \cite{Agafonov68} is another striking result relating
normality and automata.  It establishes that oblivious selection of symbols
with a regular set preserves normality.  Oblivious selection with a regular
set~$L$ means that a symbol~$a_i$ is selected whenever the prefix
$a_1 \cdots a_{i-1}$ belongs to~$L$.  This oblivious selections can
actually be realized by sequential transducers considered in this work.

The paper is organized as follows.  Notions of normal sequences and
transducers are introduced in Section~\ref{sec:basic}.  Main results are
stated in Section~\ref{sec:results}.  Proofs of the results and
algorithms are given in Section~\ref{sec:proofs}.

\section{Basic Definitions} \label{sec:basic}

\subsection{Normality}

Before giving the formal definition of normality, let us introduce some
simple definitions and notation. Let $A$ be a finite set of symbols that we
refer to as the \emph{alphabet}. We write $A^{\omega}$ for the set of all
infinite words on the alphabet~$A$ and $A^*$ for the set of all finite
words.  The length of a finite word $w$ is denoted by $|w|$.  The positions
of finite and infinite words are numbered starting from~$1$. To denote the
symbol at position $i$ of a word $w$ we write $w[i]$, and to denote the
substring of $w$ from position $i$ to $j$ inclusive we write $w[i \hdots
j]$.  The empty word is denoted by~$\lambda$.  The cardinality of a finite
set~$E$ is denoted by~$\#E$.

Given two words $w$ and $v$ in $A^*$, the number $|w|_v$ of
occurrences of $v$ in $w$ is defined by:
\begin{displaymath}
  |w|_v = \# \{ i : w[i \hdots i+\lvert v \rvert - 1] = v \} .
\end{displaymath}
For example, $|abbab|_{ab} = 2$.

Given a finite word $w \in A^+$ and an infinite word $x \in A^{\omega}$, we
refer to the \emph{frequency of $w$ in $x$} as
\begin{displaymath}
   \freq(x,w) = \lim_{n \rightarrow \infty} \frac{\size{x[1 \dots n]}_{w}}{n}
\end{displaymath}
when this limit is well-defined.

An infinite word $x \in A^{\omega}$ is \emph{normal} on the alphabet $A$
if for every word $w \in A^*$:
\begin{displaymath}
  \freq(x,w) = \frac{1}{(\sizeSet{A})^{|w|}}
\end{displaymath}

An alternative definition of normality can be given by counting
\emph{aligned} occurrences, and it is well-known that they are equivalent
(see for example \cite{BecherCarton18}).  We refer the reader to
\cite[Chap.4]{Bugeaud12} for a complete introduction to normality.

The most famous example of a normal word is due to
Champernowne~\cite{Champernowne33}, who showed in 1933 that the infinite
word obtained from concatenating all the natural numbers (in their usual
order):
\begin{displaymath}
  0123456789101112131415161718192021222324252627282930\hdots
\end{displaymath}
is normal on the alphabet $\{0,1,\ldots,9\}$.

\subsection{Input-deterministic Transducers}

In this paper we consider automata with outputs, also known as transducers.
Such finite-state machines are used to realize functions mapping words to
words and especially infinite words to infinite words.  We only consider
input-deterministic transducers, also known as \emph{sequential} in the
literature (See \cite[Sec.~V.1.2]{Sakarovitch09} and
\cite[Chap.~IV]{Berstel79}).  Each transition of these transducers consumes
exactly one symbol of their input and outputs a finite word which might be
empty.  Furthermore, ignoring the output label of each transition must
yield a deterministic automaton.

More formally a \emph{transducer} $\transducer$ is a tuple $\langle
Q,A,B,\delta,q_0 \rangle$, where $Q$ is a finite set of states, $A$ and $B$
are the input and output alphabets respectively, $\delta \subseteq Q \times
A \times B^* \times Q$ is a finite transition relation and $q_0 \in Q$ is
the initial state.  A transition is a tuple $\langle p,a,v,q \rangle$ in $Q
\times A \times B^* \times Q$ and it is written $p \trans{a|v} q$.

\begin{figure}[htbp]
  \begin{center}
  \begin{tikzpicture}
  \node (qinit) at (-1,1) {} ;
  \node[style={circle,draw}] (q1) at (0,0) {$1$} ;
  \node[style={circle,draw}] (q2) at (2,-2) {$2$} ;
  \node[style={circle,draw}] (q3) at (2,0) {$3$} ;
  \path [->] (qinit) edge [bend left] (q1) ;
  \path [->, loop left] (q1) edge node {$a | a$} (q1) ;
  \path [->]             (q1) edge [bend left] node {\hspace{-.5cm}\raisebox{-.75cm}{$b | \lambda$}} (q2) ;
  \path [->]             (q2) edge node {\hspace{.75cm}{$a | \lambda$}} (q3) ;
  \path [->]             (q2) edge [bend left] node {\hspace{-1cm}\raisebox{-.75cm}{$b | bb$}} (q1) ;
  \path [->, loop right] (q3) edge node {$a | \lambda$} (q3) ;
  \path [->]             (q3) edge [right] node[xshift=-7pt] {\raisebox{.5cm}{$b | ba$}} (q1) ;
  \end{tikzpicture}
  \caption{A deterministic complete transducer}
  \label{fig:transducer}
  \end{center}
\end{figure}
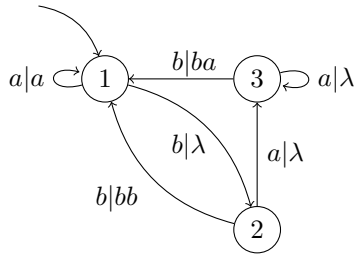

A transducer $\transducer$ is \emph{input-deterministic}, or
\emph{sequential} for short, if whenever $p \trans{a | v} q$ and
$p \trans{a \mid v'} q'$ are two of its transitions, then $q=q'$ and
$v=v'$.  A transducer $\transducer$ is \emph{complete} if for each symbol
$a \in A$ and each state $p \in Q$ there is a transition from $p$ and
consuming $a$, that is, there exists a transition $p \trans{a | v} q$.

A finite (respectively infinite) \emph{run} in $\transducer$ is a finite
(respectively infinite) sequence of consecutive transitions,
\begin{displaymath}
   q_0 \trans{a_1|v_1} q_1
       \trans{a_2|v_2}
       \cdots
       q_{n-1}
       \trans{a_n|v_n} q_n
\end{displaymath}
Its \emph{input} and \emph{output labels} are the words $a_1 a_2 \ldots
a_n$ and $v_1v_2\cdots v_n$ respectively.  Note that there is no accepting
condition and note also that the output label of an infinite run might be
finite since the output label of some transitions might be empty.  An
infinite run is \emph{accepting} if its first state is initial and its
output label is infinite.  If $\mathcal{T}$ is an input-deterministic
transducer, each infinite word~$x$ is the input label of at most one
accepting run in~$\mathcal{T}$.  When this run does exist, its output is
denoted by~$\mathcal{T}(x)$.

Each transducer~$\mathcal{T}$ can be seen as a graph by ignoring the labels
of its transitions.  For this reason, we may consider \emph{strongly
  connected components} (SCC) of~$\mathcal{T}$.  Using the terminology of
Markov chains, a strongly connected component is called \emph{recurrent} if
no transition leaves it.

We say that a input-deterministic transducer $\transducer$ \emph{preserves
  normality} if for each normal word~$x$, $\transducer(x)$ is also normal.

\subsection{Weighted Automata}

We now introduce weighted automata.  In this paper we only consider
weighted automata whose weights are rational numbers with the usual
addition and multiplication (See \cite[Chap.~III]{Sakarovitch09} for 
a complete introduction).

A \emph{weighted automaton} $\automaton$ is a tuple $\langle Q,B,\Delta,I,F
\rangle$, where $Q$ is the state set, $B$ is the alphabet,
$I:Q\rightarrow \mathbb{Q}$ and $F:Q\rightarrow \mathbb{Q}$ are the
functions that assign to each state an initial and a final weight and
$\Delta:Q \times B \times Q \rightarrow \mathbb{Q}$ is a function that
assigns to each transition a weight.

As usual, the weight of a run is the product of the weights of its
transitions times the initial weight of its first state and times the final
weight of its last state.  Furthermore, the weight of a word $w \in B^*$ is
the sum of the weights of all runs with label~$w$.

\begin{figure}[htbp]
  \begin{center}
  \begin{tikzpicture}
  \node (qinitq0) at (-1,0) {} ;
  \node (qfinq1) at (3,0) {} ;
  \node[style={circle,draw}] (q0) at (0,0) {$q_0$} ;
  \node[style={circle,draw}] (q1) at (2,0) {$q_1$} ;
  \path [->] (qinitq0) edge node {\raisebox{.6cm}{\R{1}}} (q0);
  \path [->] (q1) edge node {\raisebox{.6cm}{\A{1}}} (qfinq1);
  \path [->] (q0) edge node {\raisebox{.65cm}{$1$\V{$\WG{1}$}}} (q1) ;
  \path [->, loop above] (q0) edge node {$\begin{array}{c} 1\V{\WG{1}} \\ 0\V{\WG{1}} \end{array}$} (q0) ;
  \path [->, loop above] (q1) edge node {$\begin{array}{c} 0\V{\WG{2}} \\ 1\V{\WG{2}} \end{array}$} (q1) ;
  \end{tikzpicture}
  \end{center}
  \caption{A weighted automaton}
  \label{fig:weighted} 
\end{figure}
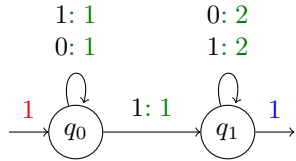

A weighted automaton is pictured in Figure~\ref{fig:weighted}.  A
transition $p \trans{a} q$ with weight $x$ is pictured $p \trans{a\WG{x}}
q$.  Non-zero initial and final weights are given over small incoming and
outgoing arrows.  The weight of the run $q_0 \trans{1} q_1 \trans{0} q_1
\trans{1} q_1 \trans{0} q_1$ is $1 \cdot 1 \cdot 2 \cdot 2 \cdot 2 \cdot 1
= 8$.  The weight of the word $w = 1010$ is $8 + 2 = 10$.  More generally
the weight of a word~$w = a_1 \cdots a_k$ is the integer~$n = \sum_{i=1}^k
{a_i2^{k-i}}$ ($w$ is a binary expansion of~$n$ with possibly some leading
zeros).

\section{Results} \label{sec:results}

We now state the main results of the paper.  The first one states that when
a transducer is strongly connected, deterministic and complete, the
frequency of each finite word~$w$ in the output of a run with a normal
input label is given by a weighted automaton over~$\mathbb{Q}$.  The second
one states that it can be checked in cubic time whether an
input-deterministic transducer preserves normality.

\begin{theorem}\label{thm:WAinCubicTime}
  Given a transducer $\transducer = \langle Q,A,B,\delta,q_0 \rangle$ which
  is strongly connected, deterministic and complete, there exists a weighted
  automaton $\automaton$ such that for any normal word~$x$ and for any
  finite word $w$, $\freq(\mathcal{T}(x),w) = \weight_{\automaton}(w)$.
  
  Furthermore, the weighted automaton~$\automaton$ can be computed in cubic
  time with respect to the size of the transducer~$\transducer$.
\end{theorem}

The hypothesis that the transducer is complete guarantees that each normal
word~$x$ is the input label of an infinite run. Nevertheless, this run may
not be accepting since its output label is not necessarily infinite.  The
really restrictive hypothesis is that the transducer must be
input-deterministic.

To illustrate Theorem~\ref{thm:WAinCubicTime} we give in
Figure~\ref{fig:weighted2} a weighted automaton~$\mathcal{A}$ which
computes the frequency of each finite word~$w$ in $\mathcal{T}(x)$ for a
normal input~$x$ and the transducer~$\mathcal{T}$ pictured in
Figure~\ref{fig:transducer}.  States $2$ and~$3$ are useless and could be
removed since their initial weight is zero and they have no incoming
transitions.  They have been kept because the automaton pictured in
Figure~\ref{fig:weighted2} is the one given by the procedure described in
the next section.

\begin{figure}[htbp]
  \begin{center}
  \begin{tikzpicture}
    \node (qinit1) at (-1,1) {} ;
    \node (qinit4) at (-1,-2) {} ;
    \node (qinit5) at (2,1.5) {} ;
    \node (qfin2) at (3,-2) {} ;
    \node (qfin3) at (3,0) {} ;
    \node[style={circle,draw}] (q1) at (0,0) {$1$} ;
    \node[style={circle,draw}] (q2) at (2,-2) {$2$} ;
    \node[style={circle,draw}] (q3) at (2,0) {$3$} ;
    \node[style={circle,draw}] (q4) at (0,-2) {$4$} ;
    \node[style={circle,draw}] (q5) at (1,1.5) {$5$} ;
    \path [->] (qinit1) edge node {\hspace{-.75cm}\raisebox{-0.45cm}{\R{$2/3$}}} (q1) ;
    \path [->] (qinit4) edge node {\raisebox{-0.75cm}{\R{$1/6$}}} (q4) ;
    \path [->] (qinit5) edge node {\raisebox{-0.75cm}{\R{$1/6$}}} (q5) ;
    \path [transform canvas={yshift=0.5ex, xshift=0.7ex}, ->] (q1) edge node {\raisebox{0.55cm}{\A{$1$}}} (qinit1) ;
    \path [transform canvas={yshift=0.7ex}, ->] (q4) edge node {\raisebox{0.35cm}{\A{$1$}}}  (qinit4) ;
    \path [transform canvas={yshift=0.7ex}, ->] (q5) edge node {\raisebox{0.35cm}{\A{$1$}}} (qinit5) ;
    \path [->] (q2) edge node {\raisebox{0.35cm}{\A{$1$}}} (qfin2) ;
    \path [->] (q3) edge node {\raisebox{0.35cm}{\A{$1$}}} (qfin3) ;
    \path [->, loop left]                        (q1) edge node {$a\V{\WG{1/2}}$} (q1) ;
    \path [transform canvas={xshift=-0.6ex}, ->] (q1) edge node {\hspace{-1.25cm}{$b\V{\WG{1/4}}$}} (q4) ;
    \path [->, bend left]                        (q1) edge node {\hspace{-.5cm}\raisebox{1.05cm}{$b\V{\WG{1/4}}$}} (q5) ;
    \path [->]                                   (q2) edge node {\raisebox{0.40cm}{$b\V{\WG{1/2}}$}} (q4) ;
    \path [->]                                   (q2) edge node {\hspace{1.5cm}\raisebox{-1.5cm}{$b\V{\WG{1/2}}$}} (q5) ;
    \path [->]                                   (q3) edge node {\hspace{0.85cm}{$b\V{\WG{1}}$}}   (q5) ;
    \path [transform canvas={xshift=0.6ex}, ->]  (q4) edge node {\hspace{0.75cm}$b\V{\WG{1}}$}   (q1) ;
    \path [->, bend left]                        (q5) edge node {\raisebox{-1.05cm}{$a\V{\WG{1}}$}}   (q1) ;
  \end{tikzpicture}
  \end{center}
  \caption{A weighted automaton for the transducer pictured in Fig.~\ref{fig:transducer}}
  \label{fig:weighted2}
\end{figure}
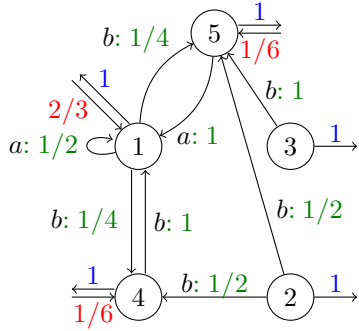

\begin{theorem}\label{thm:determinePreservation}
  Given a transducer $\transducer = \langle Q,A,B,\delta,q_0 \rangle$ which
  is complete and deterministic, it can be decided in cubic time with
  respect to the size of~$\mathcal{T}$ whether $\transducer$ preserves
  normality.
\end{theorem}

\section{Algorithms and Proofs} \label{sec:proofs}

In this section we provide the proofs for Theorems \ref{thm:WAinCubicTime}
and~\ref{thm:determinePreservation}.  The next proposition shows that it
suffices to independently analyze each recurrent strongly connected
component of the transducer.
\begin{proposition}\label{prp:decomposition}
  A deterministic and complete transducer preserves normality if and only
  if each of its recurrent strongly connected components preserves
  normality.
\end{proposition}

The previous proposition follows directly from the next lemma
which is Satz~2.5 in~\cite{Schnorr71}.
\begin{lemma}
  A run labeled with a normal word in a deterministic and complete automaton
  always reaches a recurrent strongly connected component.
\end{lemma}

By Proposition~\ref{prp:decomposition} it suffices to analyze preservation
of normality in each recurrent strongly connected component.  In what
follows we mainly consider strongly connected transducers.  If all
transitions have an empty output label, the output of any run is empty and
the transducer does not preserve normality.  Therefore, we assume that
transducers have at least one transition with a non empty output label.  By
Lemmas \ref{lem:statefreq} and~\ref{lem:runfreq}, this transition is
visited infinitely often if the input is normal because all entries of the
stationary distribution are positive \cite[Thm~1.1(b)]{Senata06}.  This
guarantees that each normal word~$x$ is the input label of an accepting run
and that $\mathcal{T}(x)$ is well-defined.

Some frequencies are obtained as stationary distributions of Markov chains
\cite[Thm~4.1]{Bremaud08}.  For that purpose, we associate a Markov chain
$\mathcal{M}$ to each strongly connected automaton~$\mathcal{A}$.  For
simplicity, we assume that the state set~$Q$ of~$\mathcal{A}$ is the set
$\{1, \ldots,\#Q\}$.  The state set of the Markov chain is the same set
$\{1, \ldots,\#Q\}$.  The transition matrix of the Markov chain is the
matrix $P = (p_{i,j})_{1\le i,j \le \#Q}$ where each entry~$p_{i,j}$ is
equal to $\#\{ a: i \trans{a} j\}/\#A$.  Note that $\#\{ a: i \trans{a}
j\}$ is the number of transitions from~$i$ to~$j$.  Since the automaton is
assumed to be deterministic and complete, the matrix~$P$ is stochastic.  If
the automaton~$\mathcal{A}$ is strongly connected, the Markov chain is
irreducible and it has therefore a unique stationary distribution~$\pi$
such that $\pi P = \pi$.  Note that all entries of $P$ and~$\pi$ are
rational numbers which can be effectively computed from~$\mathcal{A}$. The
vector~$\pi$ is called the \emph{distribution} of~$\mathcal{A}$.  This
definition as well as Lemmas \ref{lem:statefreq} and~\ref{lem:runfreq}
below also apply to input-deterministic transducers by ignoring the output
labels.

Each run of either an automaton or a transducer can be seen as a sequence
of transitions.  Therefore, the notion of the \emph{frequency} $\freq(\rho,
\gamma)$ of a finite run~$\gamma$ in an infinite run~$\rho$ is defined as
in Section~\ref{sec:basic}.  Note that $\freq(\rho, \gamma)$ is a limit and
might not exist.  This notion applies also to states seen as runs of
length~$0$.

The following lemma which is Lemma~4.5 in~\cite{Schnorr71} states that
if the automaton~$\mathcal{A}$ is strongly connected, then the run on
a normal input visits each state with a frequency.  Moreover, these
frequencies are independent of the input as long as it is normal.

\begin{lemma} \label{lem:statefreq}
  Let $\mathcal{A}$ be a deterministic and complete automaton which is
  strongly connected and let $\rho$ be a run in~$\mathcal{A}$ labeled by
  a normal word.  Then the frequency $\freq(\rho,q)$ is equal to~$\pi_q$
  for each state~$q$ where $\pi$ is the stationary distribution of the
  Markov chain associated to~$\mathcal{A}$.
\end{lemma}

Let $\gamma$ be a finite run whose first state is~$p$ and let $\rho$ be an
infinite run.  We call \emph{conditional frequency} of~$\gamma$ in~$\rho$
the ratio $\freq(\rho,\gamma)/\freq(\rho,p)$.  It is defined as soon as
both frequencies $\freq(\rho,\gamma)$ and $\freq(\rho,p)$ do exist.

\begin{lemma} \label{lem:runfreq}
  Let $\mathcal{A}$ be a deterministic and complete automaton which is
  strongly connected and let $\rho$ be a run in~$\mathcal{A}$ labeled by
  a normal word.  The conditional frequency of each run of length~$k$ is
  $1/(\#A)^k$.
\end{lemma}
\begin{proof}
  Let $p * w$ denote the unique run starting in state~$p$ with label~$w$.
  We now define an automaton whose states are the runs of length~$n$
  in~$\automaton$.  We let $\automaton^n$ denote the automaton whose state
  set is $\{ p * w : p \in Q, w \in A^n\}$ and whose set of transitions is
  defined by
  \begin{displaymath}
    \left\{
      {(p * bw) \textstyle\trans{a} (q * wa)} :
      {p \textstyle\trans{b} q}\text{ in $\automaton$}, \;\; a,b \in A
      \text{ and }w \in A^{n-1}
    \right\}
  \end{displaymath}
  The Markov chain associated with the automaton~$\automaton^n$ is called
  the \emph{snake} Markov chain.  See Problems 2.2.4, 2.4.6 and 2.5.2
  (page~90) in \cite{Bremaud08} for more details.  It is pure routine to
  check that the stationary distribution~$\xi$ of~$\automaton^n$ is given
  by $\xi_{p * w} = \pi_p/(\sizeSet{A})^n$ for each state~$p$ and each
  word~$w$ of length~$n$ and where $\pi$ is the stationary distribution
  of~$\mathcal{A}$.  To prove the statement, apply
  Lemma~\ref{lem:statefreq} to the automaton~$\automaton^n$.
\end{proof}

The output labels of the transitions in~$\mathcal{T}$ may have arbitrary
lengths.  We first describe the construction of an equivalent
transducer~$\mathcal{T}'$ such that all output labels in~$\mathcal{T}'$
have length at most~$1$. We call this transformation \emph{normalization}
and it consists in \emph{replacing} each transition $ p \trans{a | v} q$
in~$\transducer$ such that $\size{v} \geq 2$ by $n$ transitions:
\begin{displaymath}
   p \trans{a | b_1} q_1
     \trans{\lambda | b_2} q_2
     \cdots
     q_{n-1} \trans{\lambda | b_n} q
\end{displaymath}
where $q_1, q_2, \ldots, q_{n-1}$ are new states and $v = b_1\cdots b_n$.
We refer to $p$ as the parent of $q_1, \cdots, q_{n-1}$.

To illustrate the construction, the normalized transducer obtained from
the transducer of Figure~\ref{fig:transducer} is pictured in
Figure~\ref{fig:normalized}.

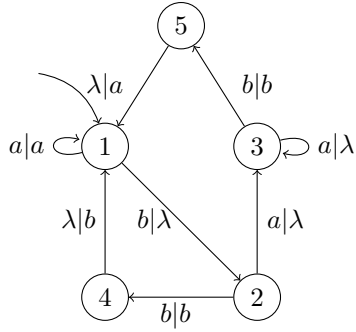
\begin{figure}[htbp]
  \begin{center}
  \begin{tikzpicture}
  \node (qinit) at (-1,1) {} ;
  \node[style={circle,draw}] (q1) at (0,0) {$1$} ;
  \node[style={circle,draw}] (q2) at (2,-2) {$2$} ;
  \node[style={circle,draw}] (q3) at (2,0) {$3$} ;
  \node[style={circle,draw}] (q4) at (0,-2) {$4$} ;
  \node[style={circle,draw}] (q5) at (1,1.6) {$5$} ;
  \path [->] (qinit) edge [bend left] (q1) ;
  \path [->, loop left] (q1) edge node {$a | a$} (q1) ;
  \path [->]             (q1) edge [left] node {\hspace{.5cm}\raisebox{0cm}{$b | \lambda$}} (q2) ;
  \path [->]             (q2) edge node {\hspace{.75cm}{$a | \lambda$}} (q3) ;
  \path [->]             (q2) edge [left] node[xshift=7pt] {\raisebox{-.75cm}{$b | b$}} (q4) ;
  \path [->]             (q4) edge [left] node {\hspace{0.0cm}\raisebox{0cm}{$\lambda | b$}} (q1) ;
  \path [->, loop right] (q3) edge node {$a | \lambda$} (q3) ;
  \path [->]             (q3) edge node {\hspace{1.0cm}{$b | b$}} (q5) ;
  \path [->]             (q5) edge node {\hspace{-1.0cm}{$\lambda | a$}} (q1) ;
  \end{tikzpicture}
  \end{center}
  \caption{The normalized transducer of the transducer pictured
           in Fig.~\ref{fig:transducer}}
  \label{fig:normalized}
\end{figure}

From the normalized transducer~$\transducer'$ we construct a weighted
automaton $\automaton$ with the same state set as $\transducer'$.  We
define transitions between every pair of states $p,q$ for each symbol~$b$
in~$B$, that is the transition $p \trans{b} q$ is defined for all states
$p,q$ and for every symbol $b\in B$.  To assign weight to transitions in
$\automaton$, we first assign weights to transitions in~$\transducer'$ as
follows.  Each transition starting from a state in~$\mathcal{T}$ (and having
a symbol as input label) has weight $1/{\sizeSet{A}}$ and each transition
starting from a newly added state (and having the empty word as input
label) has weight~$1$.  Note that for each state~$p$ in~$\transducer'$, the
sum of weights of transitions starting from~$p$ is~$1$. We now consider
separately transitions that generate empty output from those that do not.
Consider the $Q \times Q$ matrix $E$ whose $(p,q)$-entry is given for each
pair $(p,q)$ of states by
\begin{displaymath}
  E_{p,q} = \sum_{a\in A} \weight_{\transducer'}(p \trans{a|\lambda} q).
\end{displaymath}
Let $E^*$ be the matrix defined by $E^* = \sum_{k \ge 0} E^k$.  Hence the
entry $E^*_{p,q}$ is the sum of weights of all finite runs with empty
output going from~$p$ to~$q$.  The matrix~$E^*$ can be computed because it
is the solution of the system $E^* = EE^* + I$ where $I$ is the identity
matrix.  This proves in particular that all its entries are rational
numbers.

For each symbol $b\in B$ consider the $Q \times Q$ matrix $N_b$ whose
$(p,q)$-entry is given for each pair $(p,q)$ of states by
\begin{displaymath}
  (N_b)_{p,q} = \sum_{a \in A \cup \{\lambda\}}
  \weight_{\transducer'}(p \trans{a | b} q).
\end{displaymath}
We define the weight of a transition $p \trans{b} q$ in~$\mathcal{A}$ as
\begin{displaymath}
  \weight_\automaton (p \trans{b} q) = (E^*N_b)_{p,q}.
\end{displaymath}
To assign initial weight to states we consider the Markov chains
$\mathcal{M}$ whose transition matrix is the stochastic matrix $P = \sum_{b
  \in B} {E^*N_b}$. The fact that this matrix is indeed stochastic follows
from the observation that, for each state~$p$, the set of input labels of
the runs in $\bigcup_{q\in Q, b\in B}{\Gamma_{p,b,q}}$ (see below for the
definition of $\Gamma_{p,b,q}$) is a maximal prefix code and Proposition~3.8
in \cite[Chap.~II.3]{BerstelPerrin85}.  The initial vector of~$\mathcal{A}$
is the stationary distribution~$\pi$ of~$\mathcal{M}$, that is, the line
vector~$\pi$ such that $\pi P = \pi$.   We assign to each state~$q$ the
initial weight $\pi_q$.  Finally we assign final weight~$1$ to all states.

We give below the matrices $E$, $E^*$, $N_a$, $N_b$ and $P$ and the initial
vector~$\pi$ of the weighted automaton obtained from the transducer
pictured in Figure~\ref{fig:normalized}.

\begin{alignat*}{2}
  E & = 
  \begin{bmatrix} 
      0 & 1/2 & 0 & 0 & 0 \\
      0 & 0 & 1/2 & 0 & 0 \\
      0 & 0 & 1/2 & 0 & 0 \\
      0 & 0 & 0 & 0 & 0 \\
      0 & 0 & 0 & 0 & 0
  \end{bmatrix}
  &\quad\text{and}\quad
  E^* &= 
  \begin{bmatrix}
      1 & 1/2 & 1/2 & 0 & 0 \\
      0 & 1 & 1 & 0 & 0 \\
      0 & 0 & 2 & 0 & 0 \\
      0 & 0 & 0 & 1 & 0 \\
      0 & 0 & 0 & 0 & 1
  \end{bmatrix} \\
  N_a & = 
  \begin{bmatrix}
      1/2 & 0 & 0 & 0 & 0 \\
      0 & 0 & 0 & 0 & 0 \\
      0 & 0 & 0 & 0 & 0 \\
      0 & 0 & 0 & 0 & 0 \\
      1 & 0 & 0 & 0 & 0
  \end{bmatrix}
  &\quad\text{and}\quad
  N_b & = 
  \begin{bmatrix}
      0 & 0 & 0 & 0 & 0 \\
      0 & 0 & 0 & 0 & 1/2 \\
      0 & 0 & 0 & 1/2 & 0 \\
      1 & 0 & 0 & 0 & 0 \\
      0 & 0 & 0 & 0 & 0
  \end{bmatrix} \\
  P & = 
  \begin{bmatrix}
      1/2 & 0 & 0 & 1/4 & 1/4 \\
      0 & 0 & 0 & 1/2 & 1/2 \\
      0 & 0 & 0 & 0 & 1 \\
      1 & 0 & 0 & 0 & 0 \\
      1 & 0 & 0 & 0 & 0
  \end{bmatrix}
  &\quad\text{and}\quad
  \pi &= 
  \begin{bmatrix}
      2/3 & 0 & 0 & 1/6 & 1/6 
  \end{bmatrix}
\end{alignat*}

\begin{proposition} \label{prp:automatonCorrect}
  The automaton $\automaton$ computes frequencies, that is, for every normal
  word $x$ and any finite word $w$ in~$B^*$,
  $\weight_{\automaton}(w) = \freq(\mathcal{T}(x),w)$.
\end{proposition}

The proof of the proposition requires some preliminary results.  The
following lemma is straightforward and follows directly from the
normalization of~$\mathcal{T}$ into~$\mathcal{T}'$.
\begin{lemma} \label{lem:TequivTNorm}
  Both transducers $\transducer$ and $\transducer'$ realize the same
  function, that is, $\transducer(x) = \transducer'(x)$ for any infinite
  word~$x$.
\end{lemma}

Let us recall that a set of words~$L$ (resp. runs) is called
\emph{prefix-free} if no word in~$L$ is a proper prefix of another word
in~$L$. Let $\Gamma$ be a set of finite runs and $\rho$ be an infinite run.
The limit $\freq(\rho, \Gamma)$ is defined as the limit as $n$ goes to
infinity of the ratio between the number of occurrences of runs in~$\Gamma$
in the prefix of length~$n$ of~$\rho$ and~$n$.  If $\Gamma$ is prefix-free
(not to count twice the same start), the following equality holds
\begin{displaymath}
  \freq(\rho, \Gamma) = \sum_{\gamma\in\Gamma}{\freq(\rho,\gamma)}
\end{displaymath}
assuming that each term of the right-hand sum does exist.  If $\Gamma$ is a
set of finite runs starting from the same state~$p$, the \emph{conditional
  frequency} of~$\Gamma$ in a run~$\rho$ is defined as the ratio between
the frequency of~$\Gamma$ in~$\rho$ and the frequency of~$p$ in~$\rho$,
that is, $\freq(\rho, \Gamma)/\freq(\rho, p)$.  Furthermore if $\Gamma$ is
prefix-free, the conditional frequency of~$\Gamma$ is the sum of the
conditional frequencies of its elements.

Let $x$ be a fixed normal word and let $\rho$ and~$\rho'$ be respectively
the runs in~$\mathcal{T}$ and~$\mathcal{T}'$ with label~$x$.  By
Lemma~\ref{lem:statefreq}, the frequency $\freq(\rho, q)$ of each state~$q$
is $\pi_q$ where is the stationary distribution of~$\mathcal{T}$.  The
following lemma gives the frequency of states in~$\rho'$.

\begin{lemma} \label{lem:TTpfreq}
  There exists a constant~$C$ such that if $q$ is a state of~$\mathcal{T}$,
  then $\freq(\rho', q) = \freq(\rho, q)/C$ and if $q$ is newly created,
  then $\freq(\rho', q) = \freq(\rho', p)/\#A$ where $p$ is the parent
  of~$q$.
\end{lemma}
\begin{proof}
  Observe that there is a one-to-one relation between runs labeled with
  normal words in $\mathcal{T}$ and in $\mathcal{T}'$.  More precisely,
  each transition $\tau$ in~$\rho$ is replaced by $\max(1, |v_\tau|)$
  transitions in~$\rho'$ (where $v_\tau$ is the output label of~$\tau$).
  
  By combining Lemmas \ref{lem:statefreq} and~\ref{lem:runfreq}, each
  transition of~$\mathcal{T}$ has a frequency in~$\rho$.  The first result
  follows by taking
  $C = \sum_{\tau}{\freq(\rho, \tau) \cdot \max(1, |v_\tau|)}$ where the
  summation is taken over all transitions~$\tau$ of~$\mathcal{T}$ and
  $v_\tau$ is implicitly the output label of~$\tau$.  The second result
  follows from Lemma~\ref{lem:runfreq} stating that transitions have a
  conditional frequency of~$1/\#A$ in~$\rho$.
\end{proof}

For each pair $(p,q)$ of states and each symbol $b \in B$, consider the set
$\Gamma_{p,b,q}$ of runs from~$p$ to~$q$ in~$\mathcal{T}'$ that have empty
output labels for all their transitions but the last one, which has $b$ as
output label.
\begin{displaymath}
  \Gamma_{p,b,q} = \{ p \trans{a_1|\lambda} \cdots
                   \trans{a_n|\lambda} q_n \trans{a_{n+1}|b} q : n\ge 0, q_i
                   \in Q, a_i \in A \cup \{\lambda\}\} 
\end{displaymath}
and let $\Gamma$ be the union $\bigcup_{p,q\in Q, b \in
  B}{\Gamma_{p,b,q}}$.  Note that the set~$\Gamma$ is prefix-free.
Therefore, the run~$\rho'$ has a unique factorization $\rho = \gamma_0
\gamma_1 \gamma_2 \cdots$ where each $\gamma_i$ is a finite run in~$\Gamma$
and the ending state of~$\gamma_i$ is the starting state of~$\gamma_{i+1}$.
Let $(p_i)_{i\ge0}$ and $(b_i)_{i \ge 0}$ be respectively the sequence of
states and the sequence of symbols such that $\gamma_i$ belongs to
$\Gamma_{p_i,b_i,p_{i+1}}$ for each $i \ge 0$.  Let us call $\rho''$ the
sequence $p_0p_1p_2\cdots$ of states of~$\mathcal{T}'$.

\begin{lemma}
  For each state~$q$ of~$\mathcal{T}'$, the frequency $\freq(\rho'', q)$
  does exist.
\end{lemma}
\begin{proof}
  The sequence~$\rho''$ is a subsequence of the sequence of states in the
  run~$\rho'$.  An occurrence of a state~$q$ in~$\rho'$ is removed whenever
  the output of the previous transition is empty. 
  
  Consider the transducer $\hat{\mathcal{T}}$ obtained by splitting each
  state $q$ of~$\mathcal{T}$ into two states $\qEmpty$ and $\qOutput$ in
  such a way that transitions with an empty output label end in a
  state~$\qEmpty$ and other transitions end in a state~$\qOutput$.  Then
  each transition transition $p \trans{a|v} q$ is replaced by either the
  two transitions $\pEmpty \trans{a|v} \qEmpty$ and $\pOutput \trans{a|v}
  \qEmpty$ if $v$ is empty or by the two transitions $\pEmpty \trans{a|v}
  \qOutput$ and $\pOutput \trans{a|v} \qOutput$ otherwise.  The
  state~$q_0^\lambda$ becomes the new initial state and non reachable
  states are removed.  Let $\hat{\rho}$ be the run in~$\hat{\mathcal{T}}$
  labeled by~$x$.  By Lemma~\ref{lem:statefreq}, the frequencies
  $\freq(\hat{\rho}, \qEmpty)$ and $\freq(\hat{\rho}, \qOutput)$ do exist.
  Now consider the normalization $\hat{\mathcal{T}}'$
  of~$\hat{\mathcal{T}}$ and the run $\hat{\rho}'$ in~$\hat{\mathcal{T}}'$
  labeled by~$x$.  By a lemma similar to Lemma~\ref{lem:TTpfreq}, the
  frequencies $\freq(\hat{\rho}', \qEmpty)$ and $\freq(\hat{\rho}',
  \qOutput)$ do exist.  The sequence $\rho''$ is obtained
  from~$\hat{\rho}'$ by removing each occurrence of states $\qEmpty$ and
  keeping occurrences of states~$\qOutput$.  It follows that the frequency
  of each state does exist in~$\rho''$.
\end{proof}

\begin{proof}[Proof of Proposition~\ref{prp:automatonCorrect}]
  By Lemma~\ref{lem:runfreq}, the conditional frequency of each finite
  run~$\gamma$ of length~$n$ in~$\rho$ is $1/(\#A)^n$.  It follows that the
  conditional frequency of each finite run~$\gamma'$ in~$\rho'$ is equal to
  its weight in~$\transducer'$ as defined auxiliary when defining the
  weights of transitions in the weighted automaton~$\mathcal{A}$.  This
  proves that the weight of the transition $p \trans{b} q$ in~$\mathcal{A}$
  is exactly the conditional frequency of the set $\Gamma_{p,b,q}$ for each
  triple $(p,b,q)$ in $Q \times B \times Q$.  More generally, the product
  of the weights of the transitions $p_0 \trans{b_1} p_1 \cdots p_{n-1}
  \trans{b_n} p_{n}$ is equal to the conditional frequency of the set
  $\Gamma_{p_0,b_1,p_1} \cdots \Gamma_{p_n,b_n,p_{n+1}}$ in~$\rho'$.
  
  It remains to prove that the frequency of each state~$q$ in~$\rho''$ is
  indeed its initial weight in the automaton~$\mathcal{A}$.  Let us recall
  that the initial vector of~$\mathcal{A}$ is the stationary distribution
  of the stochastic matrix~$P$ whose $(p,q)$-entry is the sum $\sum_{b\in
    B}{\weight_{\mathcal{A}}(p \trans{b} q)}$, which is the conditional
  frequency of $pq$ (as a word of length~$2$) in~$\rho''$.  It follows that
  the frequencies of states in~$\rho''$ must be the stationary of the
  matrix~$P$.
  
  Since the frequency of a word $v = b_1\cdots b_n$ in $\mathcal{T}'(x)$ is
  the same as the sum over all sequences $p_0,p_1\ldots,p_{n+1}$ of the
  frequencies of $\Gamma_{p_0,b_1,p_1} \cdots \Gamma_{p_n,b_n,p_{n+1}}$
  in~$\rho'$, it is the weight of the word~$v$ in the
  automaton~$\mathcal{A}$.
\end{proof}

\begin{proof}[Proofs of Theorems \ref{thm:WAinCubicTime} and \ref{thm:determinePreservation}]
  To complete the proof of Theorems \ref{thm:WAinCubicTime} and
  \ref{thm:determinePreservation}, we exhibit an algorithm deciding in
  cubic time whether an input-deterministic transducer preserves normality.
  Let $\mathcal{T}$ be the transducer $\langle Q,A,B,\delta,\{q_0\}
  \rangle$.  By definition, its size is the sum $\sum_{\tau \in
    \delta}\size{\tau}$, where the size of a single transition $\tau = p
  \trans{a|v} q$ is~$\size{\tau} = \size{av}$.  We consider the alphabets
  to be fixed so they are not taken into account when calculating
  complexity.

  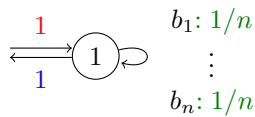
\begin{figure}[H]
    \begin{center}
      \begin{tikzpicture}
      \node[style={circle,draw}] (q6) at (0,0) {$1$} ;
      \node (qinit) at (-1.25,0) {} ;
      \path [transform canvas={yshift=0.7ex},->] (qinit) edge node {\raisebox{.6cm}{\R{1}}} (q6);
      \path [transform canvas={yshift=-0.12ex},->] (q6) edge node {\raisebox{-.8cm}{\A{1}}} (qinit);
      \path [->, loop right] (q6) edge node {\raisebox{-1cm}{$\begin{array}{c} b_1$\V{$\WG{1/n}$}$ \\ \vdots \\ b_n$\V{$\WG{1/n}$}$ \end{array}$}} (q6) ;
      \end{tikzpicture}
    \end{center}
    \caption{Weighted automaton $\mathcal{B}$
             such that
             $\weight_{\mathcal{B}}(w) = 1/(\#A)^{|w|}$}
    \label{fig:weighted3}
  \end{figure}

  By Proposition~\ref{prp:decomposition}, the algorithm decomposes the
  transducer into strongly connected components and checks that each
  recurrent one does preserve normality.  This is achieved by computing the
  weighted automaton~$\mathcal{A}$ and checking that the weight of each
  word~$w$ is $1/(\#A)^{|w|}$.  This latter step is performed by comparing
  $\mathcal{A}$ with the weighted automaton $\mathcal{B}$ such that
  $\weight_{\mathcal{B}}(w) = 1/(\#A)^{|w|}$.  The automaton~$\mathcal{B}$
  is pictured in Figure~\ref{fig:weighted3}.
\begin{itemize}
\item[]
\textbf{Input:} $\transducer = \langle Q,A,B,\delta,q_0 \rangle$ 
  an input-deterministic complete transducer.\\
\textbf{Output:} \texttt{True} if $\transducer$ preserves normality and
  \texttt{False} otherwise.\\
\textbf{Procedure:}
\begin{enumerate}
  \item[I.] Compute the strongly connected components of $\transducer$
  \item[II.] For each recurrent strongly connected component $S_i$ of $\transducer$:
    \begin{enumerate}
    \item[1.] Compute the normalized transducer $\transducer'$, equivalent to $S_i$.
    \item[2.] Use $\transducer'$ to build the weighted automaton $\automaton$:
      \begin{enumerate}
      \item[a.] Compute the weights of the transitions of~$\mathcal{A}$.
        \begin{itemize}
        \item[] Compute the matrix $E$                          
        \item[] Compute the matrix $E^*$ solving $(I-E)X=I$  
        \item[] For each $b \in B$, for each $p,q \in Q$:
          \begin{itemize}
          \item[] compute the matrix $N_b$
          \item[] define the transition $p \trans{b} q$ with
             weight $(E^*N_b)_{p,q}$.
          \end{itemize} 
        \end{itemize} 
      \item[b.] Compute the stationary distribution $\pi$ of
        the Markov chain induced by $\automaton$.
      \item[c.] Assign initial weight $\pi[i]$ to each state $i$,
        and let final weight be $1$ for all states.
      \end{enumerate}
    \item[3.] Compare $\automaton$ against the automaton $\mathcal{B}$ using
      Schützenberger's algorithm \cite{CardonCrochemore80} to check whether
      they realize the same function.
    \item[4.] If they do not compute the same function, return \texttt{False}.
    \end{enumerate}
  \item[III.] Return \texttt{True}
\end{enumerate}
\end{itemize}

Now we analyze the complexity of the algorithm.  Computing recurrent
strongly connected components can be done in $O(\size{Q}^2) \le O(n^2)$
using Kosaraju's algorithm is if the transducer is implemented with an
adjacency matrix \cite[Section~22.5]{Cormen09}.

We refer to the size of the component $\size{S_i}$ as $n_i$.  The cost of
normalizing the component is $O(n_i^2)$, mainly from filling the new
adjacency matrix.  The most expensive step when computing the transitions
and their weight is to compute $E^*$.  The cost is $O(n_i^3)$ to solve
the system of linear equations.  To compute the weights of the states we
have $O(n_i^3)$ to solve the system of equations to find the stationary
distribution.  Comparing the automaton to the one computing the expected
frequencies can be done in time $O(n_i^3)$ \cite{CardonCrochemore80} since
the coefficients of both automata are in $\mathbb{Q}$.

\end{proof}

\section*{Acknowledgements}

The authors would like to thank Ver\'onica Becher for many fruitful
discussions and suggestions.  Both authors are members of the Laboratoire
International Associ\'e INFINIS, CONICET/Universidad de Buenos
Aires--CNRS/Universit\'e Paris Diderot and they are supported by the ECOS
project PA17C04.  Carton is also partially funded by the DeLTA project
(ANR-16-CE40-0007).

\bibliographystyle{plain}
\bibliography{algorithm}

\begin{thebibliography}{10}

\bibitem{Agafonov68}
V.~N. Agafonov.
\newblock Normal sequences and finite automata.
\newblock {\em Soviet Mathematics Doklady}, 9:324--325, 1968.

\bibitem{BecherCarton18}
V.~Becher and O.~Carton.
\newblock Normal numbers and computer science.
\newblock In {\em Sequences, Groups, and Number Theory}, pages 233--269.
  Springer, 2018.

\bibitem{BecherHeiber13}
V.~Becher and P.~A. Heiber.
\newblock Normal numbers and finite automata.
\newblock {\em Theoretical Computer Science}, 477:109--116, 2013.

\bibitem{Berstel79}
J.~Berstel.
\newblock {\em Transductions and Context-Free Languages}.
\newblock B.G. Teubner, 1979.

\bibitem{BerstelPerrin85}
J.~Berstel and D.~Perrin.
\newblock {\em Theory of Codes}.
\newblock Academic Press, Inc., 1985.

\bibitem{Borel09}
{\'E}.~Borel.
\newblock Les probabilit{\'e}s d{\'e}nombrables et leurs applications
  arithm{\'e}tiques.
\newblock {\em Rend. Circ. Mat. Palermo}, 27(2):247--271, 1909.

\bibitem{Bremaud08}
P.~Br{\'e}maud.
\newblock {\em Markov Chains: Gibbs Fields, Monte Carlo Simulation, and
  Queues}.
\newblock Springer, 2008.

\bibitem{Bugeaud12}
Y.~Bugeaud.
\newblock {\em Distribution modulo one and {D}iophantine approximation}, volume
  193 of {\em Cambridge Tracts in Mathematics}.
\newblock Cambridge University Press, Cambridge, 2012.

\bibitem{CardonCrochemore80}
A.~Cardon and M.~Crochemore.
\newblock D{\'e}termination de la repr{\'e}sentation standard d'une s{\'e}rie
  reconnaissable.
\newblock {\em {ITA}}, 14(4):371--379, 1980.

\bibitem{Champernowne33}
D.~G. Champernowne.
\newblock The construction of decimals normal in the scale of ten.
\newblock {\em Journal of the London Mathematical Society}, 1(4):254--260,
  1933.

\bibitem{Cormen09}
T.~H. Cormen, Ch.~E. Leiserson, R.~L. Rivest, and C.~Stein.
\newblock {\em Introduction to Algorithms, Third Edition}.
\newblock The MIT Press, 3rd edition, 2009.

\bibitem{Dai04}
J.~Dai, J.~Lathrop, J.~Lutz, and E.~Mayordomo.
\newblock Finite-state dimension.
\newblock {\em Theoretical Computer Science}, 310:1--33, 2004.

\bibitem{Sakarovitch09}
J.~Sakarovitch.
\newblock {\em Elements of Automata Theory}.
\newblock Cambridge University Press, 2009.

\bibitem{Schnorr71}
C.~P. Schnorr and H.~Stimm.
\newblock {E}ndliche {A}utomaten und {Z}ufallsfolgen.
\newblock {\em Acta Informatica}, 1:345--359, 1972.

\bibitem{Senata06}
E.~Senata.
\newblock {\em Non-negative Matrices ans Markov Chains}.
\newblock Springer, 2006.

\end{thebibliography}

\bigskip
\bigskip

{\small
\begin{minipage}{\textwidth}
\noindent
Olivier Carton \\
Institut de Recherche en Informatique Fondamentale \\
Universit\'e Paris Diderot, France \\
Olivier.Carton@irif.fr \\
\\
Elisa Orduna\\
Universidad de Buenos Aires, Argentina\\
eorduna@dc.uba.ar
\end{minipage}
}
\end{document}